\documentclass[11pt,english]{article}
\usepackage[latin9]{inputenc}
\usepackage{geometry}
\usepackage{amsthm}
\usepackage{amsmath}
\usepackage{amssymb}

\makeatletter
 \theoremstyle{definition}
  \newtheorem{example}{\protect\examplename}
  \theoremstyle{remark}
  \newtheorem{rem}{\protect\remarkname}
\theoremstyle{plain}
\newtheorem{thm}{\protect\theoremname}
  \theoremstyle{plain}
  \newtheorem{lem}{\protect\lemmaname}
  \theoremstyle{remark}
  \newtheorem{claim}{\protect\claimname}
  \theoremstyle{definition}
  \newtheorem{defn}{\protect\definitionname}

\usepackage{complexity,color,bm}
\usepackage{colortbl}

\providecommand{\E}{\mathrm{E}}
\newcommand{\Var}{\mathrm{Var}}

\definecolor{gray-comment}{gray}{0.5}

\theoremstyle{plain}

\makeatletter
\newtheorem*{rep@theorem}{\rep@title}
\newcommand{\newreptheorem}[2]{%
\newenvironment{rep#1}[1]{%
 \def\rep@title{#2 \ref{##1}}%
 \begin{rep@theorem}}%
 {\end{rep@theorem}}}
\makeatother

\newreptheorem{rcor}{Corollary}
\newreptheorem{rthm}{Theorem}

\usepackage{thmtools}
\usepackage{thm-restate}

\DeclareMathOperator{\supp}{supp}

\usepackage{mathabx}

\newcommand{\srev}{\textsc{SRev}}

\newcommand{\brev}{\textsc{BRev}}
\newcommand{\rev}{\textsc{Rev}}
\newcommand{\prev}{\textsc{PRev}}
\newcommand{\drev}{\textsc{DRev}}

\makeatother

\usepackage{babel}
  \providecommand{\claimname}{Claim}
  \providecommand{\definitionname}{Definition}
  \providecommand{\examplename}{Example}
  \providecommand{\lemmaname}{Lemma}
  \providecommand{\remarkname}{Remark}
\providecommand{\theoremname}{Theorem}

\begin{document}

\title{{\Large{}On the Computational Complexity of Optimal Simple Mechanisms}}

\author{Aviad Rubinstein\thanks{UC Berkeley.
I am grateful to Alon Eden, Amos Fiat, and Muli Safra for inspiring discussions. I also thank Jason Hartline,  Christos Papadimitriou, Paul Tylkin, and anonymous reviewers for helpful comments on previous versions of this manuscript. Most of this research was done while the author was an intern at Microsoft Research New England. Part of this research was supported by NSF grant CCF1408635 and Templeton Foundation grant 3966.}}
\maketitle
\begin{abstract}
We consider a monopolist seller facing a single buyer with additive
valuations over $n$ heterogeneous, independent items. It is known
that in this important setting optimal mechanisms may require randomization
\cite{HartR12}, use menus of infinite size \cite{DDT15-infinite_menu},
and may be computationally intractable \cite{DaskalakisDT14}. This
has sparked recent interest in finding simple mechanisms that obtain
reasonable approximations to the optimal revenue \cite{HartN12,LiY13,BabaioffILW14}.
In this work we attempt to find the {\em optimal simple mechanism}.

There are many ways to define simple mechanisms. Here we restrict
our search to {\em partition mechanisms}, where the seller partitions
the items into disjoint bundles and posts a price for each bundle;
the buyer is allowed to buy any number of bundles. 

We give a PTAS for the problem of finding a revenue-maximizing partition
mechanism, and prove that the problem is strongly \NP-hard. En route,
we prove structural properties of near-optimal partition mechanisms
which may be of independent interest: for example, there always exists
a near-optimal partition mechanism that uses only a constant number
of non-trivial bundles (i.e. bundles with more than one item).

 \thispagestyle{empty}\setcounter{page}{0} \newpage
\end{abstract}

\section{Introduction}

Designing revenue-maximizing mechanisms for a seller who faces a single
buyer with additive valuations over $n$ heterogeneous, independent
items is a fundamental problem in auction theory. It is known that
even in this simple setting, the optimum mechanism requires randomization
\cite{HartR12}, uses a menu of infinite size \cite{DDT15-infinite_menu},
and may be computationally intractable \cite{DaskalakisDT14}. Such
mechanisms are often considered ``impractical'': buyers and sellers
may be reluctant to participate in mechanisms that are too complicated;
randomization may be restricted by legal requirements (and by our
poor understanding of risk aversion); describing and choosing among
infinite menus raises obvious issues of communication and computational
complexity, etc. Put in computer science jargon, {\bf simplicity
is a constraint}: just like the auctioneer cannot obtain the entire
social welfare because of incentive compatibility constraints, the
optimum mechanism's revenue is infeasible because of the simplicity
constraint.

In recent years there have been many works comparing simple mechanisms
to optimal mechanisms (including \cite{HartlineR09,AlaeiFHH13,Yao15,RW15-subadditive,BateniDHS15}).
In particular, a line of works \cite{HartN12,LiY13,BabaioffILW14}
in our setting (a single buyer with additive valuations over independent
items) culminated with a celebrated $1/6$-approximation of the optimal
revenue by the better of the following two mechanisms: (a) sell each
item separately; and (b) auction all the items together as one grand
bundle. Here, rather than comparing to the benchmark of the globally
optimal (but infeasible) auction, we want to find the best feasible
mechanism. Clearly, the above mechanism also obtains a $1/6$-approximation
of the optimal simple mechanism; but can we do better?

\begin{center}{\em 

Can we find the optimal simple mechanism?

}\end{center}

Alas, it is not clear how to formalize ``simple''. In this work we
propose {\em partition mechanisms} as a standard for simplicity.
(In a partition mechanisms the seller partitions the set of items
into disjoint bundles, and posts a price for each bundle; the buyer
is allowed to select any number of bundles.) In Section \ref{sec:Partition-vs-simple}
we discuss some of the reasons that made us choose this definition,
as well as some of its imperfections. We want to emphasize that the
same question could be asked with respect to any definition of ``simple''.
(For example: what is the computational complexity of finding the
optimal deterministic mechanism with polynomial (additive)-menu-size?)

Our technical contributions include a PTAS, i.e. for any constant
$\delta>0$, we give a polynomial time algorithm that finds a partition
mechanism that obtains $\left(1-\delta\right)$-approximation to the
optimal revenue among all partition mechanisms. Rather than developing
novel algorithmic techniques, our main tool is exploring the structural
properties of near-optimal partitions. For example, we prove that
there exists a near-optimal partition mechanism with only a constant
number of non-trivial bundles. We also prove that this problem is
strongly \NP-hard, i.e. there is no FPTAS (assuming $\P\neq\NP$).

\paragraph{Organization}

In Section \ref{sec:Partition-vs-simple} we discuss some of the merits
of partition mechanisms. In Section \ref{sec:Techniques,-intuition,-and}
we give a few interesting examples and provide some intuition for
the technical part. The \NP-hardness result appears in Section \ref{sec:np-hardness}
and the PTAS in Section \ref{sec:PTAS}.

\subsection{Related work}

We briefly discuss a few related works on computational complexity
of designing simple, revenue-(near)-optimal mechanisms in settings
with independent item valuations. For a constant number (or many i.i.d.)
additive buyers with monotone hazard rate (MHR) valuation distributions,
Cai and Huang \cite{CaiH13} give a PTAS to the optimal mechanism.
(Note that we make no assumption on the distributions except independence.)
Cai and Huang's mechanism is simple in the sense that most items are
sold as a single bundle, but for the remaining few items an arbitrary
(potentially randomized) mechanism is used. Our restriction to partition
mechanisms for an additive buyer has an analog of item-pricing mechanisms
for a unit-demand buyer. For the latter, Cai and Daskalakis \cite{CD11-unit_Demand-PTAS}
give a PTAS for monotone hazard rate valuations and a Quasi-PTAS for
regular valuations, and finding the exact optimum for general valuations
is \NP-hard by Chen et al. \cite{ChenDPSY14}.

\section{\label{sec:Partition-vs-simple}Partition mechanisms as simple mechanisms}

While there have been notable attempts to quantify complexity of different
mechanisms (e.g. by Hart and Nisan \cite{HartN13} and recently by
Morgenstern and Roughgarden \cite{MR15-auction_dimension}), it is
fair to say that we have not seen an indisputable, universal definition
of simple mechanisms. Most likely, because ``simple'' {\em can and
should} mean different things in different settings; for example,
compare the simplicity desiderata in the following scenarios: selling
produce in a grocery store (buyers are limited in time and computational
capacity); spectrum auctions (buyers may be limited by legal constraints);
and ad-auctions in an online marketplace (decisions are often made
by automated algorithms).

In this work we define simple mechanisms as partition mechanisms.
Certainly, there are issues with this definition. One immediate problem
with restricting to partition mechanisms is that they don't really
capture all simple mechanisms. In particular, see Example \ref{ex:prev-vs-drev}
for a distribution where a simple, deterministic mechanism that is
not a partition mechanism obtains a strictly greater revenue. More
importantly, some of the advantages of partition mechanisms listed
in this section are restricted to a single buyer, with additive valuations,
over independent items; this issue is illustrated in Example \ref{ex:prev-vs-drev-n_buyers}
which shows that for many buyers with additive valuations over independent
items, partition mechanisms achieve a revenue much lower than the
optimum. Coming up with a canonical definition for simple mechanisms
remains one of the most important open problems in this line of work.
Nevertheless, in this section we argue that partition mechanisms have
many advantages as the standard for simple mechanisms in this particular
setting.

\paragraph{Expressiveness}

We argue that despite their simplicity, partition mechanisms can be
used to express important auctions of interest. For example, they
generalize both selling items separately and bundling all the items
together; thus by \cite{BabaioffILW14} they guarantee at least a
$1/6$-approximation to the optimal revenue achievable with any mechanism.
Furthermore, this is a strong generalization: as we show in Example
\ref{ex:2-gap}, partition mechanisms can obtain as much as double
the revenue obtained by the better of selling items separately or
bundling all the items together. Also, we note partition mechanisms
can exhibit rich structure, as is evident by our \NP-hardness result.

\paragraph{Menu complexity and false-name-proofness}

Hart and Nisan \cite{HartN13} discuss a measure of menu-size complexity:
every truthful mechanism can be represented as a menu of (potentially
randomized) outcomes and prices, where the buyer is allowed to choose
one of those outcomes. As noted by Hart and Nisan, the mechanism which
auctions each item separately has exponential menu-size complexity
under this definition. To overcome this problem, they also introduce
a measure {\em additive-menu-size}, where the buyer is allowed to
buy an arbitrary number of outcomes from the menu. Under this definition,
partition auctions have linear additive-menu-size complexity.

A related issue is that of false-name-proofness, i.e. can a buyer
gain from participating in the mechanism several times? Partition
mechanisms (and additive-menu mechanisms in general) have the advantage
that they are always false-name-proof.

\paragraph{Locality and buyer-side computational complexity}

Partition mechanisms also have the advantage that the buyer's decisions
are ``local'', i.e. the decision to buy one bundle is independent
of the decision to buy other bundles. This greatly simplifies tasks
such as analyzing and reasoning about such mechanisms, learning or
predicting the effects of changes to the environment or the mechanism,
etc. In particular, this makes the buyer's decisions very easy.

\paragraph{Revenue monotonicity}

Hart and Reny \cite{HartR12} observed an interesting phenomenon they
call {\em revenue non-monotonicity}, where increasing the buyer's
valuations (in the sense of stochastic dominance), may strictly decrease
the optimal obtainable revenue. Hart and Reny showed a constant factor
gap between the revenue obtainable with the higher and lower valuations,
even when selling two i.i.d. items. Furthermore, \cite{RW15-subadditive}
recently observed that for two items with correlated valuations, this
gap may be infinite. Another nice property of partition mechanisms
is that the maximum revenue obtainable by auctions in this class is
revenue-monotone.

\section{Preliminaries}

For any distribution $\vec{D}$ of valuations, we use the following
notation, mostly due to \cite{HartN12,BabaioffILW14}, to denote the
optimum revenue for each class of mechanisms:
\begin{itemize}
\item $\rev\left(\vec{D}\right)$ - the maximum revenue among all truthful
mechanisms;
\item $\drev\left(\vec{D}\right)$ - the maximum revenue among all truthful
{\em deterministic} mechanisms;
\item $\prev\left(\vec{D}\right)$ - the maximum revenue among all truthful
{\em partition} mechanisms;
\item $\brev\left(\vec{D}\right)$ - the maximum revenue obtainable by auctioning
the grand bundle; and
\item $\srev\left(\vec{D}\right)$ - the maximum revenue obtainable by pricing
each item separately.
\end{itemize}
When $\vec{D}$ is clear from the context, we simply write $\rev,\drev,$
etc.

\section{\label{sec:Techniques,-intuition,-and}Techniques, intuition, and
examples}

We begin our technical exposition with the following example which
separates the revenue obtainable with a partition mechanism from the
better of pricing each item separately or auctioning the grand bundle.
\begin{example}
[$\prev = (2-o(1)) \max\{\srev,\brev\}$] \label{ex:2-gap}

Consider $2n$ items: 
\begin{itemize}
\item $A$: $n$ items with equal-revenue valuations. $v_{a}\in S\triangleq\left\{ 1,\dots,\sqrt{n}\right\} $,
with distribution $\Pr\left[v_{a}\geq k\right]=1/k\,\,\,\forall k\in S$;
and
\item $B$: $n$ items with rare-event valuations. $v_{b}\in\left\{ 0,\alpha n^{b}\right\} $,
with distribution $\Pr\left[v_{b}=\alpha n^{b}\right]=n^{-b}$ , where
we set $\alpha=\E\left[v_{a}\right]$. 
\end{itemize}
With a partition mechanism, we can obtain expected revenue $\left(1-o\left(1\right)\right)n\alpha$
from the items in $A$ by bundling them together, and also $n\alpha$
from $B$ by selling each item $i$ separately for price $\alpha n^{i}$.
However, selling all the items separately achieves negligible revenue
on $A$, whereas the items in $B$ will have negligible contribution
to the revenue from selling the grand bundle.\end{example}
\begin{rem}
We remark that the Example \ref{ex:2-gap} shows, in particular, a
$\left(2-o\left(1\right)\right)$-gap between $\max\left\{ \srev,\brev\right\} $
and $\rev$. Previously Babaioff et al. \cite{BabaioffILW14} cited
an example due to \cite{DaskalakisDT14} that gave a $1.05$-gap.
\end{rem}
The example above builds on the key intuition from \cite{LiY13,BabaioffILW14}
that there is an interesting tradeoff between bundling and selling
separately: when most revenue is distributed among many low-impact,
high probability events (as in subset $A$), their sum concentrates
and bundling is preferable; when most revenue comes from rare events
(as in subset $B$), we want to sell the items separately. \cite{LiY13,BabaioffILW14}
call this the {\em core-tail} decomposition. 

A nice question suggested to us by Amos Fiat is whether this is the
``only way'' that $\prev$ can beat $\max\left\{ \srev,\brev\right\} $.
In particular, is there always a revenue-maximizing partition mechanism
with at most one non-trivial bundle? The following example shows that
the answer is no.
\begin{example}
[Two non-trivial bundles] \label{ex:2-bundles}

Consider the following valuations: 
\begin{itemize}
\item for $i\in\left\{ 1,2\right\} $, let $\Pr\left[v_{i}=1\right]=\Pr\left[v_{i}=2\right]=1/2$; 
\item for $i=\left\{ 3,4\right\} $, let $\left(1/9\right)\cdot\Pr\left[v_{i}=1\right]=\Pr\left[v_{i}=10\right]=1/10$. 
\end{itemize}
The unique optimal partition mechanism offers bundle $\left\{ 1,2\right\} $
for price $3$ and bundle $\left\{ 3,4\right\} $ for price $11$.
The revenue obtained is 
\begin{gather*}
3\cdot\Pr\left[\sum_{i\in\left\{ 1,2\right\} }v_{i}\geq3\right]+11\cdot\Pr\left[\sum_{i\in\left\{ 3,4\right\} }v_{i}\geq11\right]=3\cdot\frac{3}{4}+11\cdot\frac{19}{100}=4.34.
\end{gather*}

\end{example}
The core-tail intuition from \cite{LiY13,BabaioffILW14} cannot explain
the success of the optimal partition in Example \ref{ex:2-bundles}.
For this distribution, the optimal partition exploits the fact that
the values of the bundles are slightly more likely to come out $3$
and $11$, respectively, than other values on the equal-revenue curve.
But they are still far from concentration around $3$ and $11$. Our
\NP-hardness result constructs gadgets that generalize Example \ref{ex:2-bundles}
to create instances where the optimal partition exhibits an arbitrarily
complex structure.

Our PTAS is more intricate. Let us informally sketch the main idea.
In Example \ref{ex:2-bundles}, something interesting happens at $3$,
and something interesting happens at $11$. In general, many interesting
events can happen in different locations on the (positive) real numbers
line, but one of the following two always holds:
\begin{itemize}
\item The interesting events are far apart on the real line - in this case
we don't lose much by ignoring the events that pertain to lower values.
In terms of Example \ref{ex:2-bundles}, we exploit the asymmetry
between the bundle $\left\{ 1,2\right\} $ and the bundle $\left\{ 3,4\right\} $.
\item Most of the action is restricted to a small interval - this is a redundancy
we can exploit. For example, because the sum of many independent random
variables in the same range should concentrate.
\end{itemize}
More concretely, we prove (Lemmata \ref{lem:singletons} and \ref{lem:bundles})
that there exists a near-optimal partition mechanism that uses only
a constant number of non-trivial bundles. In some sense this is a
bicreteria-approximation variant of Fiat's conjecture that Example
\ref{ex:2-gap} is the only reason we would want to use a partition
other $\max\left\{ \srev,\brev\right\} $. We then build on the same
intuition to construct modified valuations that approximate the original
distributions. Finally, we show that the new distributions admit a
succinct representation, so we can find a near-optimal partition mechanism
by brute-force search.

For completeness, let us conclude this technical exposition with two
examples that separate $\drev$ from $\prev$; they serve to remind
us that there are many interesting mechanisms beyond the scope of
partition mechanisms considered in this paper. The first example shows
a constant separation in our setting of a single additive buyer.
\begin{example}
[Hart and Nisan \cite{HartN12}; $\prev =  (1-\Omega(1))\drev$] \label{ex:prev-vs-drev}

Consider two i.i.d. items with valuations sampled uniformly from $\left\{ 0,1,2\right\} $.
The expected revenue for selling the bundle with both items (for any
price) is at most $1$; and selling each item separately (for any
price) yields total revenue at most $4/3$. Going beyond partition
mechanisms, we can offer either item for price $2$, or the grand
bundle for price $3$. The revenue obtained from this auction is
\begin{gather*}
3\cdot\Pr\left[\sum_{i\in\left\{ 1,2\right\} }v_{i}\geq3\right]+2\cdot\Pr\left[\left\{ v_{1},v_{2}\right\} =\left\{ 2,0\right\} \right]=3\cdot\frac{3}{9}+2\cdot\frac{2}{9}=13/9>4/3.
\end{gather*}

\end{example}
The second example shows that with many buyers, partition mechanisms
cannot achieve any constant fraction of the optimum revenue. See also
the recent paper by Yao \cite{Yao15} on constructing different simple
mechanisms in this setting.
\begin{example}
[e.g. \cite{BabaioffILW14}; Many buyers: $\prev =  o(1)\drev$] \label{ex:prev-vs-drev-n_buyers}

We consider $n$ items and $m=n^{1/4}$ buyers; we let $v_{i}^{j}$
denote buyer $j$'s value for item $i$. All $v_{i}^{j}$ are drawn
i.i.d. from the following distribution: with probability $1-n^{-3/4}$,
$v_{i}^{j}=0$; otherwise, $v_{j}^{i}$ is drawn from an equal-revenue
distribution with support $\left\{ 1,\dots,n^{1/10}\right\} $, i.e.
$\Pr\left[v_{i}^{j}\geq k\right]=n^{-3/4}/k$. 

For any one buyer, selling item $i$ for price $k$ yields revenue
$n^{-3/4}$, which is only an $O\left(1/\log n\right)$-fraction of
the expected value; but the total expected value for the grand bundle
concentrates, so that revenue can easily be obtained. With $m$ buyers,
the auctioneer can guarantee almost the entire social welfare with
the following mechanism: approach buyers in any order; for each buyer
charge slightly lower than her expected value for the remaining items,
and let her choose her favorite $n^{1/4}$ items. With a partition
mechanism, on the other hand, we must fix the partition without knowing
which items each buyer wants. Thus partition mechanisms can guarantee
at most an $O\left(1/\log n\right)$-fraction of the optimum revenue.
\end{example}

\section{\label{sec:np-hardness}\NP-hardness}
\begin{thm}
Given an explicit description of a product distribution of item valuations,
computing a revenue maximizing partition is strongly NP-hard.\end{thm}
\begin{proof}
We reduce from 3D-Matching: Given sets $X,Y,Z$ and a set of hyperedges
$H\subseteq X\times Y\times Z$, find a maximum 3-dimensional matching,
i.e. maximum non-intersecting subset $M\subseteq H$. Karp \cite{Karp72-NP}
proved that it is \NP-complete to decide whether there exists a perfect
matching (i.e. $\left|M\right|=\left|X\right|=\left|Y\right|=\left|Z\right|$).

\paragraph{Construction}

Identify the set of items with the set of vertices $I\triangleq X\cup Y\cup Z$.
Identify the set of hyperedges with their indices; for each $h\in H$,
let $\pi_{h}\triangleq\left|H\right|^{6}+\left|H\right|^{3}\cdot h$.
Let $\pi_{\min}\triangleq\min_{h\in H}\pi_{h}=\left|H^{6}\right|$
and $\pi_{\max}\triangleq\max_{h\in H}\pi_{h}=\left(1+O\left(1/\left|H\right|^{2}\right)\right)\left|H^{6}\right|$,
For each item $i$, the distribution of valuations $D_{i}$ is defined
by the $\pi_{h}$'s of its hyperedges. Specifically, we let 
\[
\supp\left(D_{i}\right)=\left\{ 1\right\} \cup\left\{ \pi_{h}\colon h\ni i\right\} \,\,\,\,\,\,\,\,\,\mbox{and}\,\,\,\,\,\,\,\,\,\Pr_{v_{i}\sim D_{i}}\left[v_{i}\geq\pi_{h}\right]=1/\pi_{h}.
\]

Observe that selling each item separately, for any price in its support,
yields expected revenue of $1$ per item.

\paragraph{Completeness}

If there exists a perfect matching $M$, we partition according to
this matching, and set the price for bundle $h$ at $\pi_{h}+2$.
For each bundle, we have 
\[
\Pr\left[\sum_{i\in h}v_{i}\geq\pi_{h}+2\right]=1-\left(1-1/\pi_{h}\right)^{3}=3/\pi_{h}-3/\pi_{h}^{2}+1/\pi_{h}^{3}.
\]
The expected revenue for each bundle is therefore
\[
\left(\pi_{h}+2\right)\left(3/\pi_{h}-3/\pi_{h}^{2}+1/\pi_{h}^{3}\right)=3+3/\pi_{h}-O\left(1/\pi_{h}^{2}\right)
\]
Summing over $\left|M\right|$ hyperedges (i.e. $\left|M\right|$
bundles), we guarantee a total revenue of $OPT\triangleq\left|M\right|\left(3+3/\pi_{\max}-O\left(1/\pi_{\max}^{2}\right)\right)$.

\paragraph{Soundness}

We first claim that there exists an optimum partition where every
bundle is contained in a hyperedge. Let $B$ be a bundle sold for
some price $\pi\geq\pi_{\min}$. Clearly $\pi\leq\left|B\right|\pi_{\max}$,
otherwise it never sells. Similarly, we have $\pi\leq\left|B\right|+\pi_{\max}$,
otherwise it sells with probability at most $\left|B\right|/\pi_{\min}^{2}$,
yielding revenue $\left|B\right|\pi/\pi_{\min}^{2}\ll1$.

Let $i\in B$ be such that $\pi\notin\left[\pi_{h}-\left|H\right|^{2},\pi_{h}+\left|H\right|^{2}\right]$
for all $h\ni i$. We compare the revenue from selling $B$ for price
$\pi$ to the revenue from selling $B\setminus\left\{ i\right\} $
for price $\pi-1$. If $B$ sells for price $\pi$ but $B\setminus\left\{ i\right\} $
does not sell for price $\pi-1$, then at least one of the following
must be true: (1) $v_{i}\geq\pi-\left|B\right|$, which by our assumption
on $i$ implies $v_{i}\geq\pi+\left|H\right|^{2}$; or (2) $v_{i}>1$
and there is some other $j\in B$ such that $v_{j}>1$. We bound the
probability of the union as follows:
\begin{enumerate}
\item $\Pr\left[v_{i}\geq\pi+\left|H\right|^{2}\right]\leq\frac{1}{\pi+\left|H\right|^{2}}$;
the revenue loss is bounded by $\frac{\pi}{\pi+\left|H\right|^{2}}\leq\frac{\pi-\left|H\right|^{2}/2}{\pi}$;
\item $\Pr\left[\left(v_{i}>1\right)\wedge\left(\exists j\in B\,\,\, v_{j}>1\right)\right]<\frac{\left|B\right|}{\pi_{\min}^{2}}$;
the revenue loss is bounded by $\frac{\pi\left|B\right|}{\pi_{\min}^{2}}\leq\frac{2\left|B\right|}{\pi}$. 
\end{enumerate}
There is also some revenue loss from the decrease in price: the original
bundle sells with probability at most $\frac{\left|B\right|}{\pi-\left|B\right|}\leq\frac{2\left|B\right|}{\pi}$;
since we decrease the price by $1$, $\frac{2\left|B\right|}{\pi}$
also bounds the expected loss in revenue. The total expected loss
in revenue is therefore at most $\frac{\pi-\left|H\right|^{2}/2}{\pi}+\frac{4\left|B\right|}{\pi}<1$,
so selling $B$ as a bundle cannot be optimal.

There is an optimum partition that bundles items according to hyperedges
in some partial matching $M'$, and the rest of the items are in bundles
of size at most two. The optimal price for a bundle of two items from
hyperedge $h$ is $\pi_{h}+1$; the probability of selling for this
price is $1-\left(1-1/\pi_{h}\right)^{2}=2/\pi_{h}-\pi_{h}^{2}$.
Multiplying by $\pi_{h}+1$, we get an expected revenue of $2+1/\pi_{h}$.
In particular, this is only $\left(1+1/2\pi_{h}\right)$ per item,
as opposed to $\left(1+1/\pi_{h}\right)$ per item with a full hyperedge. 
\end{proof}

\section{\label{sec:PTAS}PTAS}
\begin{thm}
For any constant $\delta>0$, there exists a deterministic polynomial
time algorithm that, given an explicit description of a product distribution
of item valuations, computes a partition and prices that generate
a $\left(1-\delta\right)$-approximation to the maximum revenue obtainable
by partition mechanisms.
\end{thm}

\paragraph{Proof outline}

In the next two subsections we prove a structural characterization
of near optimal auctions: Lemmata \ref{lem:singletons} and \ref{lem:bundles}
imply that there exists a near-optimal partition mechanism that uses
only a constant number of non-trivial bundles (i.e. bundles with more
than one item). Furthermore, the prices to these bundles are all within
a constant factor, and all these bundles sell for these prices with
constant probability.

In Subsection \ref{sub:Discretization} we use our insight about the
structure of near-optimal auctions to show that for optimizing over
this restricted class of partitions, most of the information in the
distribution is redundant. In particular we can place every item in
one of $O\left(\log n\right)$ buckets, where the items within each
bucket are indistinguishable for the algorithm. For each bucket, there
are constantly-many options to approximately partition the identical
items among constantly many bundles (or to be sold separately). We
can enumerate over all approximate partitions for all buckets in polynomial
time. 

See also description of the algorithm in Subsection \ref{sub:Algorithm}.

\subsection{Singletons}

Given the following lemma, we can assume wlog that every non-trivial
bundle sells with probability at least $\epsilon$.
\begin{lem}
\label{lem:singletons}For any $\delta>0$, let $\epsilon\leq\delta^{3}/4$.
Let $B\subseteq\left[n\right]$ be a bundle of items, and let $\pi_{B}\in\mathbb{R}_{+}$
be an optimal price for $B$. Suppose that the revenue from auctioning
$B$ for price $\pi_{B}$ is $\epsilon\cdot\pi_{B}$, i.e. $\Pr\left[\sum_{i\in B}v_{i}\geq\pi_{B}\right]=\epsilon$.
Then the revenue from selling the items in $B$ separately is at least
$\left(1-\delta\right)\cdot\epsilon\cdot\pi_{B}$.\end{lem}
\begin{proof}
For the proof of this lemma, we simplify notation by normalizing to
$\pi_{B}=1$.

Below we prove that most of the revenue comes from the item with the
highest value (this may be a different item in each realization).
In particular, if the total value of the bundle is at least $1$,
then it is likely that there is a single item whose value is almost
$1$,
\begin{equation}
\Pr\left[\max_{i\in B}v_{i}\geq1-\delta/2\mid\sum_{i\in B}v_{i}\geq1\right]\geq1-\delta/2.\label{eq:sum-max}
\end{equation}
This means in particular that $\Pr\left[\max_{i\in B}v_{i}\geq1-\delta/2\right]\geq\left(1-\delta/2\right)\epsilon$,
and therefore selling each item separately for price $\left(1-\delta/2\right)$
guarantees a $\left(1-\delta/2\right)^{2}\geq\left(1-\delta\right)$-fraction
of the revenue from selling $B$ as a bundle.

We now prove (\ref{eq:sum-max}). Since $1$ is an optimal price for
$B$, we have that the $\Pr\left[\sum_{i\in B}v_{i}\geq\delta/4\right]\leq\epsilon/\left(\delta/4\right)$,
otherwise $\delta/4$ would have been a better price. What is the
probability that there exist a partition $B=S\cup T$ such that $\sum_{i\in S}v_{i}\geq\delta/2$
and $\sum_{i\in T}v_{i}\geq\delta/2$? If we were to fix any partition
$B=U\cup V$ before observing the realizations, or to pick one uniformly
at random, we would have 
\[
\Pr\left[\left(\sum_{i\in U}v_{i}\geq\delta/4\right)\wedge\left(\sum_{i\in V}v_{i}\geq\delta/4\right)\right]\leq\left(\epsilon/\left(\delta/4\right)\right)^{2}.
\]
Now assume that there exist some partition $\left(S,T\right)$ as
above, and pick $\left(U,V\right)$ uniformly at random. With probability
at least $1/4$ we have that $\sum_{i\in\left(S\cap V\right)}v_{i}\geq\sum_{i\in\left(S\cap U\right)}v_{i}$
and $\sum_{i\in\left(T\cap V\right)}v_{i}\leq\sum_{i\in\left(T\cap U\right)}v_{i}$;
and the same for the event that $\sum_{i\in\left(S\cap V\right)}v_{i}\geq\sum_{i\in\left(S\cap U\right)}v_{i}$
and $\sum_{i\in\left(T\cap V\right)}v_{i}\leq\sum_{i\in\left(T\cap U\right)}v_{i}$.
Thus with probability at least $1/2$, 
\[
\min\left\{ \sum_{i\in U}v_{i},\sum_{i\in V}v_{i}\right\} \geq\min\left\{ \sum_{i\in S}v_{i},\sum_{i\in T}v_{i}\right\} /2.
\]
Therefore the probability that there exist such $\left(S,T\right)$
is at most $2\left(\epsilon/\left(\delta/4\right)\right)^{2}$.

Observe that whenever $\sum_{i\in B}v_{i}\geq1$ and $\max_{i\in B}v_{i}<1-\delta/2$,
there exists a partition $\left(S,T\right)$ as above: Let $S=\left\{ \arg\max_{i\in B}v_{i}\right\} $
and $T=B\setminus S$; if $\max_{i\in B}v_{i}\geq\delta/2$, we're
done. Otherwise, move items from $T$ to $S$ until $\sum_{i\in S}v_{i}\geq\delta/2$;
since the last item we moved from $S$ to $T$ had value at most $\max_{i\in B}v_{i}<\delta/2$,
we have $\sum_{i\in T}v_{i}\geq1-\delta>\delta/2$. Therefore, 
\[
\Pr\left[\max_{i\in B}v_{i}<1-\delta/2\mid\sum_{i\in B}v_{i}\geq1\right]\leq\frac{\Pr\left[\left(\max_{i\in B}v_{i}<1-\delta/2\right)\wedge\left(\sum_{i\in B}v_{i}\geq1\right)\right]}{\Pr\left[\sum_{i\in B}v_{i}\geq1\right]}\leq\frac{32\epsilon}{\delta^{2}}.
\]
Plugging in $\delta=4\epsilon^{1/3}$ yields (\ref{eq:sum-max}). 
\end{proof}

\subsection{Bundles}
\begin{lem}
\label{lem:bundles}For any constants $0<\epsilon\leq1/2$ and $\delta>0$
we can replace all the bundles that sell with probability at least
$\epsilon$ with $\ell=\poly\left(1/\epsilon,1/\delta\right)$ bundles,
while maintaining a $\left(1-2\delta\right)$-fraction of the expected
revenue.\end{lem}
\begin{proof}
In Claim \ref{cla:few-bundles} we show that we can recursively combine
bundles until in any interval of multiplicative-constant-length $\left[\epsilon\cdot\pi,\pi\right]$,
there is at most a constant ($k=8\epsilon^{-4}\delta^{-3}$) number
of bundles. Then, in Claim \ref{cla:constant-range} we show that
we can ignore all bundles except those in some slightly larger interval
$\left[\eta\cdot\pi,\pi\right]$ (for $\eta=\delta^{4}\epsilon^{5}\left(1-\epsilon\right)$).
This is a union of $\log_{\epsilon}\eta\leq\frac{\log\delta^{4}\epsilon^{4}\left(1-\epsilon\right)}{\log\epsilon}$
smaller intervals $\left[\epsilon\cdot\pi,\pi\right]$; together with
the sparsity we obtained in Claim \ref{cla:few-bundles}, this implies
that we are left with at most $\ell=k\log_{\epsilon}\eta$ bundles. \end{proof}
\begin{claim}
\label{cla:few-bundles}For any $\epsilon,\delta>0$, let $k=8\epsilon^{-4}\delta^{-3}$.
Consider only bundles $B_{i}$ that sell for price $\pi_{i}$ with
probability at least $\epsilon$. Partition the positive reals into
multiplicative intervals $\left[\epsilon\cdot\pi,\pi\right]$. Consider
$k$ separate bundles $B_{1},\dots,B_{k}$ with prices $\pi_{1},\dots,\pi_{k}$
in the same interval $\left[\epsilon\cdot\pi,\pi\right]$, and associated
probabilities of selling $p_{1},\dots,p_{k}$. Whenever we encounter
such a $k$-tuple, we combine them into one bundle $B'$ with price
$\pi'=\sum_{i}p_{i}\cdot\pi_{i}$ and probability $p'=1$. Recurse
until every interval has at most $k$ bundles. (The number of bundles
decreases at each step, so this process is guaranteed to terminate.)
Finally, discount all newly formed bundles by a factor of $\left(1-\delta/2\right)$. 

Then every newly formed bundle sells with probability at least $1-\delta/2$.
In particular, this guarantees an $\left(1-\delta\right)$ approximation
to the original revenue.\end{claim}
\begin{proof}
Let $B'=\bigcup B_{i}$ be any newly formed bundle, where $B_{i}$'s,
$\pi_{i}$'s, and $p_{i}$'s are the original bundles, prices and
probabilities (i.e. $B'$ may denote a union of union of bundles).
Let $v_{B_{i}}$ denote the random variable $v_{B_{i}}\triangleq\sum_{j\in B_{i}}v_{j}$;
let also $\hat{v}_{B_{i}}\triangleq\min\left\{ v_{B_{i}},\pi_{i}\right\} $.
Then we have 
\[
\E\left[\sum_{i}\hat{v}_{B_{i}}\right]\geq\sum_{i}p_{i}\cdot\pi_{i}\,\,\,\,\,\,\mbox{and}\,\,\,\,\,\:\Var\left[\sum_{i}\hat{v}_{B_{i}}\right]\leq\sum_{i}\pi_{i}^{2}.
\]

Applying Chebyshev's inequality, 
\[
\Pr\left[\sum_{i}\hat{v}_{B_{i}}\leq\left(1-\delta/2\right)\sum_{i}p_{i}\cdot\pi_{i}\right]\leq\frac{\sum\pi_{i}^{2}}{\left(\frac{\delta}{2}\cdot\sum_{i}p_{i}\cdot\pi_{i}\right)^{2}}.
\]
In the last union that formed $B'$, we combined at least $k$ bundles,
each with $p_{i}\pi_{i}\in\left[\epsilon^{2}\pi,\pi\right]$. Therefore,
\[
\frac{\sum_{i\in\left[k\right]}\pi_{i}^{2}}{\left(\frac{\delta}{2}\cdot\sum_{i\in\left[k\right]}p_{i}\cdot\pi_{i}\right)^{2}}\leq\frac{\sum_{i\in\left[k\right]}\pi^{2}}{\left(\frac{\delta}{2}\cdot\sum_{i\in\left[k\right]}\epsilon^{2}\pi\right)^{2}}\leq\frac{4}{\delta^{2}k\epsilon^{4}}.
\]
Plugging in $k=8\epsilon^{-4}\delta^{-3}\geq-\ln\left(\delta/2\right)\cdot(4\epsilon^{-4}\delta^{-2})$
guarantees that we sell the grand bundle with probability at least
$1-\delta/2$.\end{proof}
\begin{claim}
\label{cla:constant-range}For any $\epsilon,\delta>0$, and let $\eta=\delta^{4}\epsilon^{5}\left(1-\epsilon\right)$.
Let bundles $B_{1},\dots,B_{m}$, have optimal prices $\pi_{1},\dots,\pi_{m}$,
and denote $\pi^{*}=\max_{i\in\left[m\right]}\pi_{i}$. Suppose that
bundle $B_{i}$ sells for price $\pi_{i}$ with probability $p_{i}\geq\epsilon$,
for every $i\in\left[m\right]$. Suppose further that in each range
$\left[\epsilon\pi,\pi\right]$ of prices we have at most $k=\epsilon^{-4}\delta^{-3}$
bundles (this is wlog by the previous claim). Then a $\left(1-\delta\right)$-fraction
of the revenue can be obtained by selling only the bundles with with
price $\pi_{i}\geq\eta\pi^{*}$.\end{claim}
\begin{proof}
$k$ bundles with prices in interval $\left[\epsilon\pi,\pi\right]$
can yield at most $k\pi$ revenue. Summing over $\pi\in\left\{ \eta\pi^{*},\epsilon\cdot\eta\pi^{*},\epsilon^{2}\cdot\eta\pi^{*},\dots\right\} $,
we have that all those bundles together yield revenue at most $k\eta\pi^{*}/\left(1-\epsilon\right)$.
Plugging in $\eta=\delta^{4}\epsilon^{5}\left(1-\epsilon\right)$
completes the proof of the claim.
\end{proof}

\subsection{\label{sub:Discretization}Discretization}

In this section we consider a sequence of manipulations on the distribution
of each item's valuations. At the end of the manipulation, every item
will fit in one of $O\left(\log n\right)$ buckets, with all the items
in each bucket having indistinguishable distributions. The first step
is to discretize the valuation distributions:
\begin{defn}
\label{def:rounding}Let $\vec{D}\triangleq\bigtimes D_{i}$ be a
valuation distribution over non-negative reals ($\mathbb{R}_{+}$).
Let ${\cal N}_{\epsilon}\triangleq\left\{ 0\right\} \cup\left\{ \dots,\left(1+\epsilon\right)^{-1},1,\left(1+\epsilon\right),\left(1+\epsilon\right)^{2},\dots\right\} $
be a multiplicative-$\left(1+\epsilon\right)$-net over $\mathbb{R}_{+}$.
For each $i$, we construct the {\em rounded valuation distribution}
$D_{i}^{\left(1\right)}$ as follows: (a) round down every valuation
in the support to the nearest smaller (or equal) element in ${\cal N}_{\epsilon}$;
then (b) round down every probability of valuation in the new support
to the nearest smaller (or equal) element in ${\cal N}_{\epsilon}$.
Finally, we let $\vec{D}^{\left(1\right)}\triangleq\bigtimes D_{i}^{\left(1\right)}$.
\end{defn}
The following lemma implies that for sufficiently small $\epsilon>0$,
the loss in revenue from rounding the valuations is negligible.
\begin{lem}
\label{lem:rounding}For any constant $\epsilon>0$, non-negative
product distribution $\vec{D}$, and price $\pi$, we have
\begin{equation}
\Pr_{\vec{v}^{\left(1\right)}\sim\vec{D}^{\left(1\right)}}\left[\sum v_{i}^{\left(1\right)}\geq\left(1-\delta\right)\pi\right]\geq\left(1-\delta\right)\Pr_{\vec{v}\sim\vec{D}}\left[\sum v_{i}\geq\pi\right],\label{eq:rounding}
\end{equation}
where $\delta\triangleq2\epsilon^{1/3}$.\end{lem}
\begin{proof}
Rounding the valuations to ${\cal N}_{\epsilon}$ can decrease the
sum by a factor of at most $\left(1+\epsilon\right)$. Rounding the
probabilities is slightly trickier. An equivalent way of formulating
the rounded valuation distribution is to sample $\vec{v}\sim\vec{D}$,
round down the valuation of each item to ${\cal N}_{\epsilon}$, and
then zero the valuation of each item independently with probability
at most $\epsilon$. Inequality (\ref{eq:rounding}) now follows from
\[
\Pr_{\vec{v}\sim\vec{D}}\left[\sum v_{i}^{\left(1\right)}\geq\left(1-\delta\right)\pi\mid\sum v_{i}\geq\pi\right]\geq1-\delta.
\]
In particular, it suffices to show that for every $\vec{v}$ such
that $\sum v_{i}\geq\pi$, 
\begin{equation}
\Pr_{\vec{v}^{\left(1\right)}}\left[\sum v_{i}^{\left(1\right)}\geq\left(1-\delta\right)\pi\mid\vec{v}\right]\geq1-\delta,\label{eq:roudning2}
\end{equation}
 where the randomness is only over the independent zeroing of each
valuation.

Fix any such $\vec{v}$ and let $\pi_{\vec{v}}\triangleq\sum v_{i}\geq\pi$.
The expectation of the sum is at least $\E\left[\sum v_{i}^{\left(1\right)}\right]\geq\left(1-\epsilon\right)^{2}\pi_{\vec{v}}$,
and the variance is at most $\epsilon\cdot\left(1-\epsilon\right)^{2}\pi_{\vec{v}}^{2}$.
Therefore by Chebyshev's inequality, 
\[
\Pr\left[\sum v_{i}^{\left(1\right)}\leq\left(1-\delta\right)\pi_{\vec{v}}\right]\leq\frac{\epsilon}{\left(\delta/2\right)^{2}}.
\]
Plugging in $\delta=2\epsilon^{1/3}$ completes yields (\ref{eq:roudning2}).
\end{proof}
Recall that we can assume wlog that all our bundles sell with constant
probability for prices in $\left[\epsilon\pi,\pi\right]$ (Lemma \ref{lem:bundles}).
Thus, for the purpose of (approximately) evaluating an item's contribution
to any bundle it suffices to consider only its valuations in $\left[\pi/n^{2},\pi\right]$,
rounding down larger valuations to $\pi$ and ignoring smaller. Notice
that the new support has size at most logarithmic: $\left|{\cal N}_{\epsilon}\cap\left[\pi/n^{2},\pi\right]\right|=O\left(\log n\right)$.
Similarly, for the purpose of bundling, we can assume wlog that every
valuation in the support has probability at least $1/n^{2}$. Now
each value in the support is associated with one of $\left|{\cal N}_{\epsilon}\cap\left[1/n^{2},1\right]\right|=O\left(\log n\right)$
potential probabilities. 

In order to represent each item we need to know one more number -
the revenue it can generate when sold separately. Here again we can
assume wlog that this revenue is in $\left[\pi/n^{2},\pi\right]$:
if it is less than $\pi/n^{2}$, its revenue is negligible and we
never want to sell this item separately; if it is greater than $\pi$,
we always want to sell this item separately. The revenue from selling
an item separately is a product of two numbers (price and probability)
in ${\cal N}_{\epsilon}$, and therefore also belongs to ${\cal N}_{\epsilon}$.
As before, this means that we can assume wlog that the expected revenue
takes one of $\left|{\cal N}_{\epsilon}\cap\left[\pi/n^{2},\pi\right]\right|=O\left(\log n\right)$
values.

So far for each item we need $O\left(\log n\right)$ numbers, each
from a set of size $O\left(\log n\right)$. While this is much more
succinct than the naive representation, it is still not good enough
for our algorithmic application (at this point we still need $\left(\log n\right)^{O\left(\log n\right)}\gg n$
buckets). In the next two steps we reduce to only three numbers from
sets of size $O\left(\log n\right)$: first, we show that in the lower
end of the support it suffices to keep the aggregate expectation rather
than probability of each value; second, we argue that we can assume
wlog that all the high values in the support have approximately the
same probability; and the third number is the expected revenue from
selling separately.

\subsubsection*{Low values}

Fix any $\epsilon>0$, distribution $D_{i}^{\left(1\right)}$ over
$\mathbb{R}_{+}$, and $\pi\in\mathbb{R}_{+}$. Define 
\[
v_{i}^{\left(2\right)}\triangleq\begin{cases}
v_{i}^{\left(1\right)} & v_{i}^{\left(1\right)}\geq\epsilon\pi\\
\E_{u\sim D_{i}^{\left(1\right)}}\left[u\mid u<\epsilon\pi\right] & v_{i}^{\left(1\right)}<\epsilon\pi
\end{cases};
\]
round down the new value and probability to the nearest smaller elements
in ${\cal N}_{\epsilon}$, and let $\vec{D}_{i}^{\left(2\right)}$
denote the resulting new distribution.

Intuitively, it may be helpful to think of $\vec{D}^{\left(2\right)}$
as ``erasing'' or ``blurring'' the information about $\vec{D}^{\left(1\right)}$
below $\epsilon\pi$. Notice that for any $\epsilon>0$, $\E_{v_{i}^{\left(2\right)}\sim D_{i}^{\left(2\right)}}\left[v_{i}^{\left(2\right)}\right]=\E_{v_{i}^{\left(1\right)}\sim D_{i}^{\left(1\right)}}\left[v_{i}^{\left(1\right)}\right]$.
The next lemma shows that $\vec{D}^{\left(2\right)}$ also generates
approximately the same revenue. 
\begin{lem}
\label{lem:low-values}Let $\epsilon>0$, let $\delta=2\epsilon^{1/3}$,
and let $\vec{D}^{\left(1\right)}$ be a product distribution. Then,
\[
\Pr_{\vec{v}^{\left(1\right)}\sim\vec{D}^{\left(1\right)}}\left[\sum v_{i}^{\left(1\right)}\geq\pi\right]-2\delta\leq\Pr_{\vec{v}^{\left(2\right)}\sim\vec{D}^{\left(2\right)}}\left[\sum v_{i}^{\left(2\right)}\geq\left(1-\delta\right)\pi\right]-\delta\leq\Pr_{\vec{v}^{\left(1\right)}\sim\vec{D}^{\left(1\right)}}\left[\sum v_{i}^{\left(1\right)}\geq\left(1-2\delta\right)\pi\right].
\]

\end{lem}
Recall that by Lemma \ref{lem:singletons}, we can assume wlog that
all bundles sell with probability $\Omega\left(1\right)$; thus we
can tolerate the above additive loss in probability.
\begin{proof}
By Chebyshev's inequality, the sum of valuations less than $\epsilon\pi$
is within an additive $\pm\delta\pi$ of its expectation with probability
at least $1-\delta$ (and the rest of the valuations don't change
at all).
\end{proof}

\subsubsection*{High values}

For any $\epsilon>0$, a rounded distribution $D_{i}^{\left(2\right)}$
(in the sense of Definition \ref{def:rounding}), and $\pi\in\mathbb{R}_{+}$,
let $p_{i}^{*}\triangleq\max_{v\in\left[\epsilon\pi,\pi\right]}\Pr_{u\sim D_{i}^{\left(2\right)}}\left[u=v\right]$
denote the most likely valuation in $\left[\epsilon\pi,\pi\right]$.
Let $L_{D_{i}^{\left(2\right)}}^{\epsilon}$ denote the set of unlikely
high valuations: 
\[
L_{D_{i}^{\left(2\right)}}^{\epsilon}\triangleq\left\{ v\colon\Pr_{u\sim D_{i}^{\left(2\right)}}\left[u=v\right]<\epsilon^{4}p_{i}^{*}\right\} \cap\left[\epsilon\pi,\pi\right]
\]
Let $D_{i}^{\left(3\right)}$ denote the restriction of $D_{i}^{\left(2\right)}$
to $\supp\left\{ D_{i}^{\left(2\right)}\right\} \setminus L_{D_{i}^{\left(2\right)}}^{\epsilon}$:
\[
v_{i}^{\left(3\right)}\triangleq\begin{cases}
v_{i}^{\left(2\right)} & v_{i}^{\left(2\right)}\notin L_{D_{i}^{\left(2\right)}}^{\epsilon}\\
0 & v_{i}^{\left(2\right)}\in L_{D_{i}^{\left(2\right)}}^{\epsilon}
\end{cases}.
\]

The following lemma implies that for each $i$ it suffices to maintain
$p_{i}^{*}$ (which always takes one of $\left|{\cal N}_{\epsilon}\cap\left[1/n^{2},1\right]\right|=O\left(\log n\right)$
values), and a constant number of bits for each of the (constantly
many) values in $\supp\left\{ D_{i}^{\left(3\right)}\right\} \cap\left[\epsilon\pi,\pi\right]$.
\begin{lem}
\label{lem:high-values}Let $\vec{D}^{\left(2\right)}$ be a rounded
product distribution of valuations of items in $B$. Assume that selling
bundle $B$ for price $\pi$ yields higher expected revenue than selling
all the items in $B$ separately. Then,
\[
\Pr_{\vec{v}^{\left(3\right)}\sim\vec{D}^{\left(3\right)}}\left[\sum_{i\in B}v_{i}^{\left(3\right)}\geq\pi\right]\geq\left(1-\delta\right)\Pr_{\vec{v}^{\left(2\right)}\sim\vec{D}^{\left(2\right)}}\left[\sum_{i\in B}v_{i}^{\left(2\right)}\geq\pi\right].
\]
\end{lem}
\begin{proof}
We compare, for each item $i$, the potential revenue loss from switching
from $D_{i}^{\left(2\right)}$ to $D_{i}^{\left(3\right)}$ to the
expected revenue from selling $i$ separately. We argue that for each
$i$, the latter is much larger. Summing over all items, we have that
the total loss is much smaller than the total revenue from selling
every item separately. By the premise, the latter is less than the
original revenue from selling the bundle.

By definition, item $i$ has probability at least $p_{i}^{*}$ of
having value $\epsilon\pi$. Thus we can obtain revenue at least $p_{i}^{*}\cdot\epsilon\pi$
from selling it for price $\epsilon\pi$. In contrast, every time
we zero a valuation, we could potentially lose revenue $\pi$. The
total probability over items in $L_{D_{i}^{\left(2\right)}}^{\epsilon}$
is $\left|L_{D_{i}^{\left(2\right)}}^{\epsilon}\right|\epsilon^{4}p_{i}^{*}$.
Between $\epsilon\pi$ and $\pi$, there are at most $-\log\epsilon/\log\left(1+\epsilon\right)<1/\epsilon^{2}$
elements in ${\cal N}_{\epsilon}$; in particular, $\left|L_{D_{i}^{\left(2\right)}}^{\epsilon}\right|<1/\epsilon^{2}$.
Therefore, the revenue lost is at most an $\epsilon$-fraction of
the revenue from selling $i$ separately. 
\end{proof}

\subsubsection*{From QPTAS to PTAS}

We have reduced the representation of each item to three numbers in
${\cal N_{\epsilon}}\cap\left[\pi/n^{2},\pi\right]$: the expectation
over lower values, $\E_{u\sim D_{i}^{\left(3\right)}}\left[u\mid u<\epsilon\pi\right]$
(rounded down to ${\cal N}_{\epsilon}$); the quantity $p_{i}^{*}\cdot\pi$,
where $p_{i}^{*}\triangleq\max_{v\in\left[\epsilon\pi,\pi\right]}\Pr_{u\sim D_{i}^{\left(3\right)}}\left[u=v\right]$
is the maximum probabilities over higher values; and the maximum revenue
from selling $i$ separately, $\srev\left(D_{i}\right)$. (For the
higher values we also need a constant number of bits to specify which
values have probabilities close to $p_{i}^{*}$, and how close.) At
this point we need $O\left(\log^{3}n\right)$ buckets, which would
suffice for obtaining a Quasi-PTAS.

Our final step is to observe that if any of those three numbers is
much lower than the maximum of the three, it might as well be zero.
If $\E_{u\sim D_{i}}\left[u\mid u<\epsilon\pi\right]$ is much higher
than $p_{i}^{*}\cdot\pi$, then it's contribution to revenue for any
bundle outweighs the contribution from any of the higher values appearing
with very low probability; similarly if $\E_{u\sim D_{i}}\left[u\mid u<\epsilon\pi\right]\gg\srev\left(D_{i}\right)$
we would always sell item $i$ as part of one of the bundles. If $\srev\left(D_{i}^{\epsilon}\right)$
is much higher than either of the other two, then the contribution
from selling item $i$ separately outweighs the contribution (from
the lower values, higher values, or both) to the revenue from any
bundle. Finally, since $\srev\left(D_{i}\right)\geq p_{i}^{*}\epsilon\pi$,
the revenue from selling separately is never much lower than the higher
values' contribution.

\subsection{Algorithm\label{sub:Algorithm}}

We achieve a $\left(1-\delta\right)$-approximation of the optimal
partition revenue for some constant $\delta>0$; let $\epsilon=\epsilon\left(\delta\right)>0$
be a sufficiently small constant, and let 
\[
{\cal N}_{\epsilon}\triangleq\left\{ 0\right\} \cup\left\{ \dots,\left(1+\epsilon\right)^{-1},1,\left(1+\epsilon\right),\left(1+\epsilon\right)^{2},\dots\right\} .
\]
For each item $i$, compute the optimum expected revenue from selling
$i$ separately, $\srev\left(D_{i}\right)$.

Before we analyze bundles, we first want to guess a range $\left[\epsilon\pi,\pi\right]$
in which all the bundle prices will lie. Let $v^{\min}$ denote the
minimum over all nonzero values in the support of all items, and let
$v^{\max}$ denote the sum, over all items, of the maximal values
in their supports. An optimal $\pi$ must belong to $\left[v^{\min},v^{\max}\right]$.
Enumerate over all potential $\pi$'s in ${\cal N}_{\epsilon}$. By
Lemma \ref{lem:bundles} for some choice of $\pi$, it suffices to
optimize only over partitions with a constant number of non-trivial
bundles, and each of those bundles sells for prices in $\left[\epsilon\pi,\pi\right]$
with probability at least $\epsilon$. For the rest of the algorithm
assume we have such an optimal choice of $\pi$. 

For each $i$, round $D_{i}$ as in Definition \ref{def:rounding},
and let $\vec{D}^{\left(1\right)}$ denote the resulting product distribution.
($\vec{D}^{\left(1\right)}$ is stochastically dominated by $\vec{D}$,
thus the revenue obtained from a partition mechanism with valuations
drawn from  $\vec{D}^{\left(1\right)}$ is at most the revenue obtained
with the same partition and pricing with valuations drawn from $\vec{D}$.
In the other direction, Lemma \ref{lem:rounding} guarantees that
the revenue lost is at most a small constant fraction.)

For each $i$, replace all values in $\left[0,\epsilon^{2}\pi\right]$
with their expectation, and let $\vec{D}^{\left(2\right)}$ denote
the resulting product distribution. (By Lemma \ref{lem:low-values},
optimizing over bundles with prices in $\left[\epsilon\pi,\pi\right]$
with valuations drawn from $\vec{D}^{\left(2\right)}$ is the same
up to a small constant factor as with valuations drawn from $\vec{D}^{\left(1\right)}$.)

For each $i$, let $p_{i}^{*}$ denote the maximal probability $D_{i}^{\left(2\right)}$
gives to any value in $\left[\epsilon^{2}\pi,\pi\right]$. Remove
values with probabilities much smaller than $p_{i}^{*}$ from the
support of $D_{i}^{\left(2\right)}$, and let $\vec{D}^{\left(3\right)}$
denote the resulting product distribution. ($\vec{D}^{\left(3\right)}$
is stochastically dominated by $\vec{D}^{\left(2\right)}$, thus the
revenue obtained from a partition mechanism with valuations drawn
from $\vec{D}^{\left(3\right)}$is at most the revenue obtained with
the same partition and pricing with valuations drawn from $\vec{D}^{\left(2\right)}$.
In the other direction, Lemma \ref{lem:high-values} guarantees that
the revenue lost is at most a small constant fraction.)

For each $i$, we now have three variables which may be of different
scale: $\srev\left(D_{i}\right)$, $\E_{u\sim D_{i}^{\left(3\right)}}\left[u\mid u<\epsilon^{2}\pi\right]$,
and $p_{i}^{*}\cdot\pi$; we also have the full description of $D_{i}^{\left(3\right)}$
restricted to $\left[\epsilon^{2}\pi,\pi\right]$, which given $p_{i}^{*}$
requires only a constant number of bits. If any of the three variables
is much smaller than any of the others, set the smaller variable to
zero. 

We now represent each item with a constant number of variables,  which
are all either zero or within constant factors. In total, we have
at most $O\left(\log n\right)$ distinct representations, which we
henceforth call {\em buckets}.

Enumerate over the number of bundles (by Lemma \ref{lem:bundles}
it suffices to consider only numbers up to some constant $\ell$).
For each bucket, we must decide how many items to allocate to each
of the $\ell$ bundles, and which to sell separately. I.e. we must
pick some vector in $\left[0,1\right]^{\ell+1}$, and up to $\pm\epsilon$,
there are at most $\epsilon^{-\left(\ell+1\right)}$ different vectors;
in particular, only a constant number of choices for each bucket.
Enumerate (in polynomial time) over all choices for all $O\left(\log n\right)$
buckets.

\bibliographystyle{plain}
\bibliography{bib}

\end{document}